\documentclass[preprint]{elsarticle}
\usepackage{amsmath}
\usepackage[utf8]{inputenc}
\usepackage{hyperref}
\usepackage{url}
\usepackage{listings}
\usepackage{verbatim}
\usepackage{graphicx}
\usepackage{amssymb}
\usepackage{amsthm}
\usepackage{xspace}
\usepackage{float}

\newcommand\subfig[2]{{\fig{#1}{#2}}}
\newcommand\subcap[1]{{(#1):}}
\newcommand{\vvactive}{\vvec}
\newcommand{\vactive}{v}
\newcommand{\static}{static\xspace}
\newcommand{\stalled}{stalled\xspace}
\newcommand{\nbr}{{\text{\~{\j}}}}
\newcommand{\thread}{\iota}
\newcommand{\target}{j}

\newcommand{\threada}{a}
\newcommand{\threadb}{b}
\newcommand{\distance}{\text{distance}}
\newcommand{\buffer}[1]{\ensuremath{\{#1 \}}}
\newcommand{\iteration}{I}

\newcommand{\Lref}{\ensuremath{\LCAL^{\text{ref}}}}

\newcommand{\SET}[1]{\{#1\}}

\newcommand{\eq}[1]{eq.~\eqref{#1}}

\newcommand{\fig}[1]{Fig.~\ref{#1}}
\newcommand{\quot}[1]{``#1''}

\newcommand{\sect}[1]{Section~\ref{#1}} 
\newcommand{\alg}[1]{Algorithm~\ref{#1}} 
\newcommand{\algg}[1]{Alg.~\ref{#1}} 
\newcommand{\algtwo}[2]{Algorithms~\ref{#1} and \ref{#2}} 
\newcommand{\alggtwo}[2]{Algs~\ref{#1} and \ref{#2}} 
 
\newcommand{\algthree}[3]{Algorithms~\ref{#1}, \ref{#2}, and \ref{#3}} 
 
\newcommand{\rem}[1]{Remark~\ref{#1}} 
\newcommand{\remtwo}[2]{Remarks~\ref{#1} and~\ref{#2}}

\newcommand{\ACAL}{\mathcal{A}}  
\newcommand{\CCAL}{\mathcal{C}}  
\newcommand{\GCAL}{\mathcal{G}}  
\newcommand{\LCAL}{\mathcal{L}}  
\newcommand{\OCAL}{\mathcal{O}}  
\newcommand{\tautilde}{\tilde{\tau}}

\newcommand{\expc}[1]{\exp \glc #1 \grc} 
\newcommand{\loga}[2][]{\log^{#1}\! \gla #2 \gra}  

\newcommand{\minb}[2][]{\min^{#1} \glb #2 \grb}  

\newcommand{\gla}{\,}  
\newcommand{\gra}{}  
\newcommand{\glb}{\left(}  
\newcommand{\grb}{\right)}  
\newcommand{\glc}{\left[}  
\newcommand{\grc}{\right]}  
\newcommand{\gle}{\left|}  
\newcommand{\gre}{\right|}  

\newcommand{\TO}{,\ldots,}
\newcommand{\VEC}[1]{\mathbf{#1}}
\newcommand{\vvec}{\VEC{v}}

\newcommand{\xvec}{\VEC{x}}

\newcommand{\ttilde}{\tilde{t}}

\newcommand\bigOb[1]{\ensuremath{\OCAL\glb #1 \grb}}

\newcommand\diff[1]{\mathrm{d}#1}

\DeclareMathOperator{\ran}{ran}

\newcommand{\fpn}[2]{\ensuremath{#1 \! \times \! 10^{#2}}}

\newcommand{\draftfigure}[4][\linewidth]{\begin{figure}[htbp]
   \begin{center}
      \includegraphics[width=#1]{#2}
   \end{center}
   \caption{#3}
   \label{#4}
\end{figure}}

\newcommand{\Gthree}[1][]{\ensuremath{\GCAL^{(3)}_{#1}}\xspace}
\newcommand{\Gmin}[1][]{\ensuremath{\GCAL^{\min}_{#1}}\xspace}
\newcommand{\arrow}[2]{\ensuremath{[#1\!\!  \to\! \!#2]}}
\newcommand{\lift}[1]{\ensuremath{ l_{#1}}}
\newcommand{\lifted}[2]{\ensuremath{ \SET{#1, #2}}}
\newcommand{\lifting}[5]{\ensuremath{\lift{#5}=(\arrow{#1}{#2},(#3,#4),#5) }}
\newcommand{\impact}[2]{$\langle #1,#2 \rangle$}
\newcommand{\xCPU}{x86 CPU\xspace}
\newcommand{\ACPU}{ARM CPU\xspace}
\newcommand{\xCPUs}{x86 CPUs\xspace}
\newcommand{\ACPUs}{ARM CPUs\xspace}
\newcommand{\kbirth}{k'}
\newcommand{\Nbirth}{N'}

\newtheorem{algo}{Algorithm}
\newtheorem{lemma}{Lemma}
\newtheorem{remark}{Remark}
\newtheorem*{availability}{Code availability}

\begin{document}
  \title{Multithreaded event-chain Monte Carlo with local times}
  \author[LPENS]{Botao Li}
  \author[UT,ISSP]{Synge Todo} 
  \author[ESPCI]{A. C. Maggs}
  \author[LPENS]{Werner~Krauth\corref{vvv}}
  \cortext[vvv]{Corresponding author, email address:
  \texttt{werner.krauth@ens.fr} }

\address[LPENS]{Laboratoire de Physique de l’Ecole normale supérieure, ENS, 
Université PSL, CNRS, Sorbonne Université, Université Paris-Diderot, Sorbonne 
Paris Cité, Paris, France}

\address[UT]{Department of Physics, University of Tokyo, 113-0033 Tokyo, Japan}
\address[ISSP]{Institute for Solid State Physics, University of Tokyo, 277-8581 
Kashiwa, Japan}

\address[ESPCI]{CNRS UMR7083, ESPCI Paris, PSL Research University, 10 rue
Vauquelin, 75005 Paris, France}

\date{\today}

\begin{abstract}
We present a multithreaded event-chain Monte Carlo algorithm (ECMC) for hard 
spheres. Threads synchronize at infrequent breakpoints and otherwise scan for 
local horizon violations. Using a mapping onto absorbing Markov chains, we 
rigorously prove the correctness of  a sequential-consistency implementation for 
small test suites.   On x86 and ARM processors,  a C++ (OpenMP) implementation 
that uses compare-and-swap primitives for data access achieves considerable 
speed-up with respect to  single-threaded code. The generalized birthday problem 
suggests that for the number of threads scaling as the square root of the number 
of spheres, the horizon-violation probability remains small for a fixed 
simulation time. We provide C++ and Python open-source code that
reproduces all our results. 
\end{abstract}

\maketitle
\section{Introduction}
Event-chain Monte Carlo (ECMC)~\cite{Bernard2009,Michel2014JCP} is an 
event-driven realization of a continuous-time irreversible Markov chain that has 
found applications in statistical physics~\cite{Bernard2011,Kapfer2015PRL} and 
related fields~\cite{Hasenbusch_2018}. Initially restricted to hard spheres and 
to models with piece-wise constant pair potentials~\cite{Bernard2012b}, ECMC was 
subsequently extended to continuous potentials, such as spin models and all-atom 
particle systems  with long-range interactions~\cite{KapferKrauth2017, 
Faulkner2018}. Potentials need not be pair-wise additive~\cite{Harland2017}. In 
opposition to standard Monte Carlo methods, such as the Metropolis 
algorithm~\cite{Metropolis1953}, ECMC does not evaluate the potential $U(\xvec)$ 
of a configuration $\xvec$ (nor any ratio of potentials) in order to sample the 
Boltzmann distribution $\pi(\xvec) = \expc{- \beta U(\xvec)}$, with inverse 
temperature $\beta$.

For hard spheres, ECMC is a special case of event-driven molecular 
dynamics~\cite{Alder1957,AlderWainwright1959}. In molecular dynamics, usually 
all $N$ spheres have non-zero velocities, and the number of candidate collision 
events at any time is \bigOb{N}. A central scheduler, efficiently implemented 
through a heap data structure, yields the next collision with computational 
effort \bigOb{1}, and it updates the heap in at most \bigOb{\loga{N}} 
operations~\cite{Rapaport1980,Isobe2016}. Event times are global, and the CPU 
clock advances together with the collision times. The global collision times and 
the required communication at each event complicate multithread
implementations~\cite{Lubachevsky1992,Lubachevsky1993, 
Greenberg1996,Krantz1996,Marin1997}. Domain decomposition, another strategy to 
cope with synchronization, is also problematic~\cite{Miller2004}.

In hard-sphere ECMC, a set $ \ACAL_t$ of $k<N$ \quot{active} spheres (all of 
radius $\sigma$) have the same non-zero velocity $\vvactive$ that changes 
infrequently. All other spheres are \quot{static}. At a 
lifting~\cite{Diaconis2000} \lifting{i}{j}{\xvec}{\xvec'}{t}, an active sphere 
$i$ collides at time $t$ with a target sphere $j$, at contact $|\xvec' - \xvec| 
= 2 \sigma$ (a condition that must be adapted for periodic boundary conditions). 
The lifting $\lift{t}$  connects an in-state (the configuration just before time 
$t$, at time $t^-$) with an out-state (the configuration just after time $t$, at 
time $t^+$):
\begin{equation}
\text{in-state}:
\glc
\begin{aligned}
 i \in \ACAL_{t^-}, &\quad j \not \in \ACAL_{t^-}\\
 \xvec_i(t^-) &=\xvec \\
 \xvec_j(t^-) &=\xvec' \\
 \vvec_i(t^-) &= \vvactive \\
 \vvec_j(t^-) &= 0 \\
\end{aligned}
\grc; 
\quad
\text{out-state}:  \gle
\begin{aligned}
 i \not \in \ACAL_{t^+}, &\quad j \in \ACAL_{t^+}\\
 \xvec_i(t^+) &=\xvec \\
 \xvec_j(t^+) &=\xvec' \\
 \vvec_i(t^+) &= 0 \\
 \vvec_j(t^+) &= \vvactive \\
\end{aligned}\gre.
\label{equ:CollisionRules}
\end{equation}

We consider in this paper two-dimensional spheres in a square box with periodic 
boundary conditions. In this system, the  direction of  $\vvactive$ must be 
changed at certain breakpoints for the algorithm to be 
irreducible~\cite{Levin2008}. However, we restrict our attention to ECMC in 
between two such breakpoints $h$ and $h'$  with, for concreteness, $\vvactive = 
(\vactive_x,\vactive_y)= (1,0)$. For a generic \quot{lifted}~\cite{Diaconis2000} 
initial configuration \lifted{\CCAL_h}{\ACAL_h} at $h$, ECMC is deterministic up 
to $h'$. Generically no two liftings take place at the same time $t$, so that 
they can be identified by their time.

Our multithreaded ECMC algorithm propagates $k = |\ACAL|$ active  spheres in 
independent threads, with shared memory. In between $h$ and $h'$, it only uses 
local-time attributes of each sphere. At a lifting 
\lifting{i}{j}{\xvec}{\xvec'}{t_i}, $j$ synchronizes with $i$ (the local time 
$t_j$ is set equal to $t_i$). For a sphere $i$ to move, it must not violate 
certain horizon conditions of nearby spheres $j$. In the absence of horizon 
violations between $h$ and $h'$, multithreaded ECMC will  be proven equivalent 
to the global-time process. 

The motivation for our work is twofold. First, we strive to speed up current 
hard-sphere simulations where, typically, $N \sim \fpn{1}{6}$. These simulations 
require weeks or month of run time to decorrelate from the initial 
configuration~\cite{Bernard2011,Engel2013}. Using a connection to the 
generalized birthday problem in mathematics, we will argue that such simulations 
can successfully run with  $k \lesssim \sqrt{N}$. Our approach to multithreading 
thus uses the freedom  to tune the number of active spheres. Second, by 
providing proof of concept for multithreaded ECMC algorithms, we hope to 
motivate the development of parallel ECMC algorithms for other system where 
sequential ECMC applies already.

The multithreaded ECMC algorithm is presented in two versions. One 
implementation uses the sequential-consistency model~\cite{Lamport1979}. Mapped 
onto an absorbing Markov chain, its correctness  is rigorously proven for small 
test suites. The C++ implementation uses OpenMP to map active chains onto 
hardware threads, together with atomic primitives~\cite{atomic} for fine-grained 
control of interactions between threads. Considerable speed-up with respect to a 
single-threaded version is achieved. The few simultaneously moving spheres ($k 
\ll N$) avoid communication bottlenecks between threads, even though each 
hard-sphere lifting involves only little computation.

Subtle aspects of our algorithm surface through the confrontation of the C++ 
implementation with the sequential-consistency computational model on the same 
test suites. By reordering single statements in the code, we may for example 
introduce rare bugs that are not detected during random testing, but are 
readily exhibited in the rigorous solution, and that illustrate difficulties 
stemming from possible compiler or processor re-ordering.

\begin{availability}
Cell-based ECMC for two-dimensional hard spheres is implemented (in Fortran90) 
as \verb#CellECMC.f90#. Our version is slightly modified from the original code 
written by E.~P.~Bernard (see Acknowledgements). 
The code prepares initial configurations, and it is used in 
validation scripts.
\end{availability}

\section{Algorithms: from global-time processes to multithreaded ECMC}
\label{sec:Algorithms}

In this section, we start with the definition of a continuous-time process, 
\alg{alg:GlobalContinuous}, that is manifestly equivalent to molecular dynamics 
with the collision rules of \eq{equ:CollisionRules}. Its event-driven version, 
\alg{alg:GlobalECMC}, provides the reference set \Lref of liftings used in our 
validation scripts (see \sect{sec:Implementation}). The single-threaded 
\alg{alg:LocalContinuous} relies on local times. It has correct output if no 
horizon violation takes place. Its event-driven version, \alg{alg:LocalECMC}, 
yields a practical method that can be implemented and tested. 
\algtwo{alg:MultiSeqConsist}{alg:MultiCPP} realize multithreaded ECMC, the 
latter in a highly efficient C++ implementation.

\subsection{Continuous processes and ECMC with global time}
\label{sec:GlobalContinuousECMC}

\begin{algo}[Continuous process with global time]
At global time $t=h$, an initial lifted configuration 
\lifted{\CCAL_h}{\ACAL_{h}} is given ($\vvec_i = \vvactive=(1,0)\ \forall i \in 
\ACAL_h$ and $\vvec_i = 0\ \forall i \not \in \ACAL_h$). All spheres $i$ carry 
local times $t_i$, with, initially, $t_i(h) = h\ \forall i $. For active spheres 
($i \in \ACAL_{t}$), $\diff t_i/ \diff t = 1$. At a lifting 
\lifting{i}{j}{\xvec}{\xvec'}{t} the local time of sphere $j$ is updated as 
$t_j(t^+) = t$ and, furthermore,  $\ACAL_{t^+} =  \ACAL_{t^-} \setminus \SET{i} 
\cup \SET{j}$. The algorithm stops at global time $t=h'$, and outputs the lifted 
configuration \lifted{\CCAL_{h'}}{\ACAL_{h'}}, and the set $\LCAL_{h'} = 
\SET{l_t:  h < t < h'}$ of liftings that have taken place between $h$ and $h'$. 
\label{alg:GlobalContinuous}
\end{algo}

\begin{remark}[Meaning of local times]
In \algg{alg:GlobalContinuous}, the local time $t_i(t)$ is a function of the 
global time $t$. It gives the global time at which sphere $i$ was last 
active (or $t_i(t) = h$ if $i$ was not active for $[h, t ]$). Therefore $t_i 
(t) = t \ \forall i \in \ACAL_t$ and $t_i (t) < t \ \forall i \not \in \ACAL_t$.
\end{remark}

\begin{remark}[Positivity of local-time updates]
In \algg{alg:GlobalContinuous}, at any lifting $l_t$, the update of $t_j$ is 
positive: $t_j(t^+) - t_j(t^-) > 0$.
\label{rem:PositiveUpdate}
\end{remark}

\begin{remark}[Time-reversal invariance]
\algg{alg:GlobalContinuous} is deterministic and time-reversal invariant: If an 
initial  lifted configuration \lifted{\CCAL_{h}}{\ACAL_{h}} generates the  final 
lifted configuration \lifted{\CCAL_{h'}}{\ACAL_{h'}} with $\vvactive$, then the 
latter will reproduce the initial configuration with  $-\vvactive$. The set 
$\LCAL$ of liftings is the same in both cases.
\label{rem:TimeReversalInvariance}
\end{remark}

\noindent
In order to converge towards a given probability distribution, Markov-chain 
algorithms must satisfy the global-balance condition. It states that the 
probability flow into a configuration $\CCAL$ (summed over all liftings $\ACAL$) 
must equal the probability flow out of it~\cite{Levin2008}. ECMC balances these 
flows for each lifting individually (for the uniform probability distribution).

\begin{lemma}
\algg{alg:GlobalContinuous} satisfies the global-balance condition for any 
lifted configuration \lifted{\CCAL}{\ACAL}. All lifted configurations 
accessible from a given initial configuration thus have the same statistical 
weight.
\label{lem:GlobalContinuous}
\end{lemma}

\begin{proof}
The algorithm is equivalent to molecular dynamics that conserves one-dimensional 
momenta as well as the energy. The claimed property follows for 
\algg{alg:GlobalContinuous} because it is satisfied by molecular dynamics. The 
property can be shown directly for a discretized version of 
\algg{alg:GlobalContinuous} on a rectangular grid aligned with $\vvactive$ with 
infinitesimal cell size such that each lifted configuration 
\lifted{\CCAL}{\ACAL} has a unique predecessor. The flow into each lifted 
configuration equals one. This is equivalent to global balance for the uniform 
probability distribution.
\end{proof}

\noindent The event-driven version of 
\alg{alg:GlobalContinuous} is the following:
\begin{algo}[ECMC with global time]
With input as in \algg{alg:GlobalContinuous}, in each iteration 
$\iteration=1,2,\dots$, the next global lifting time is computed as 
$t_{\iteration + 1} = t_{\iteration} + \min_{i \in \ACAL, j \not \in \ACAL} 
\tau_{ij}$,\footnote{ $\tau_{ij}$ is infinite if $i$ cannot lift with $j$ for 
the given  initial configuration $\CCAL$ and velocity $\vvactive$. The presence 
of an arrow \arrow{i}{j} in the directed constraint graph $\GCAL$ indicates that 
$\tau_{ij}$ can be finite (see \sect{sec:ConstraintGraph}).} where $\tau_{ij}$ 
is the time of flight from sphere $i$ to sphere $j$. At time $t_\iteration$, 
local times and positions of active spheres are advanced to $\ttilde = 
\minb{t_{\iteration + 1}, h'}$, and to $\xvec_i(t_{\iteration + 1}) = 
\xvec_i(t_\iteration) + (\ttilde - t_\iteration) \vvactive$, respectively. If 
$\ttilde = t_{\iteration+1}$ (a lifting $l_{\iteration + 1} = 
(\SET{i,j},\SET{\xvec, \xvec'}, t_{\iteration + 1})$ takes place), the set of 
active spheres is updated as $\ACAL_{\iteration + 1} = \ACAL_{\iteration} 
\setminus \SET{i} \cup \SET{j}$. Otherwise $\ttilde = h'$, and the algorithm 
stops. Output is as in \algg{alg:GlobalContinuous}.
\label{alg:GlobalECMC}
\end{algo}
\begin{availability}
\algg{alg:GlobalECMC} is implemented in 
\verb#GlobalTimeECMC.py# and invoked in several validation scripts, for which 
it  
generates the reference lifting sets \Lref.
\end{availability}

\noindent

\subsection{Single-threaded processes and ECMC with local times}
\label{sec:LocalContinuousECMC}

\alg{alg:LocalContinuous}, that we now describe, is a single-threaded emulation  
of our multithreaded \algtwo{alg:MultiSeqConsist}{alg:MultiCPP}. A randomly 
sampled active chain $\thread \in \{1,2,\dots\}$ advances (in what corresponds 
to a thread) for an imposed duration, at most until its local time reaches $h'$. 
 On thread $\thread$, the active sphere $i$ must remain above the horizons of 
its neighboring spheres $j$ (see \subfig{fig:HorizonConditionForward}{a}). The 
horizon condition is \begin{equation} t_i + \tau_{ij} > t_j, 
\label{equ:horizon_condition} \end{equation} where the time of flight is 
$\tau_{ij} = x_j - x_i + b_{ij}$, with $b_{ij}$ the contact separation parallel 
to $\vvactive$ between spheres $i$ and $j$. The horizon condition must be 
checked for at most three spheres $j$ for a given $i$ because  all other spheres 
are either too far for lifting with $i$ in the direction perpendicular to 
$\vvactive$ or are 
prevented from lifting with $i$ by other spheres (see 
\sect{sec:ConstraintGraph}). The algorithm aborts if a horizon violation is 
encountered. The active chain $\thread$ stops if a lifting would be to a sphere 
$j$ that is itself active. The active chain $\thread + 1$ is then started.

\begin{remark}[Double role of horizon condition]
The horizon condition of \eq{equ:horizon_condition} has two roles. First, it is 
a necessary condition for a lifting of $i$ with $j$ (if it effectively takes 
place) to produce the required positive local-time update of $t_j$ at the 
lifting time $t$ (see \rem{rem:PositiveUpdate} and 
\subfig{fig:HorizonConditionForward}{a}). Second, it is a sufficient 
non-crossing condition for any sphere $k$, ensuring that $k$ was not at a 
previous local time in conflict with $i$ (see 
\subfig{fig:HorizonConditionForward}{b}).
\label{rem:HorizonDouble}
\end{remark}

\noindent
It is for the second role discussed in \rem{rem:HorizonDouble} that the horizon 
condition is checked for all neighboring spheres $j$ of an active sphere $i$. 

\begin{remark}[False alarms from horizon condition] The horizon condition may 
lead to false alarms (see \subfig{fig:HorizonConditionForward}{b}), which could 
be avoided through the use of the non-crossing condition. The latter is more 
difficult to check, as it requires the history of past liftings. Our algorithms 
only implement the horizon condition.
\label{rem:FalseAlarms}
\end{remark}

\draftfigure[9cm]{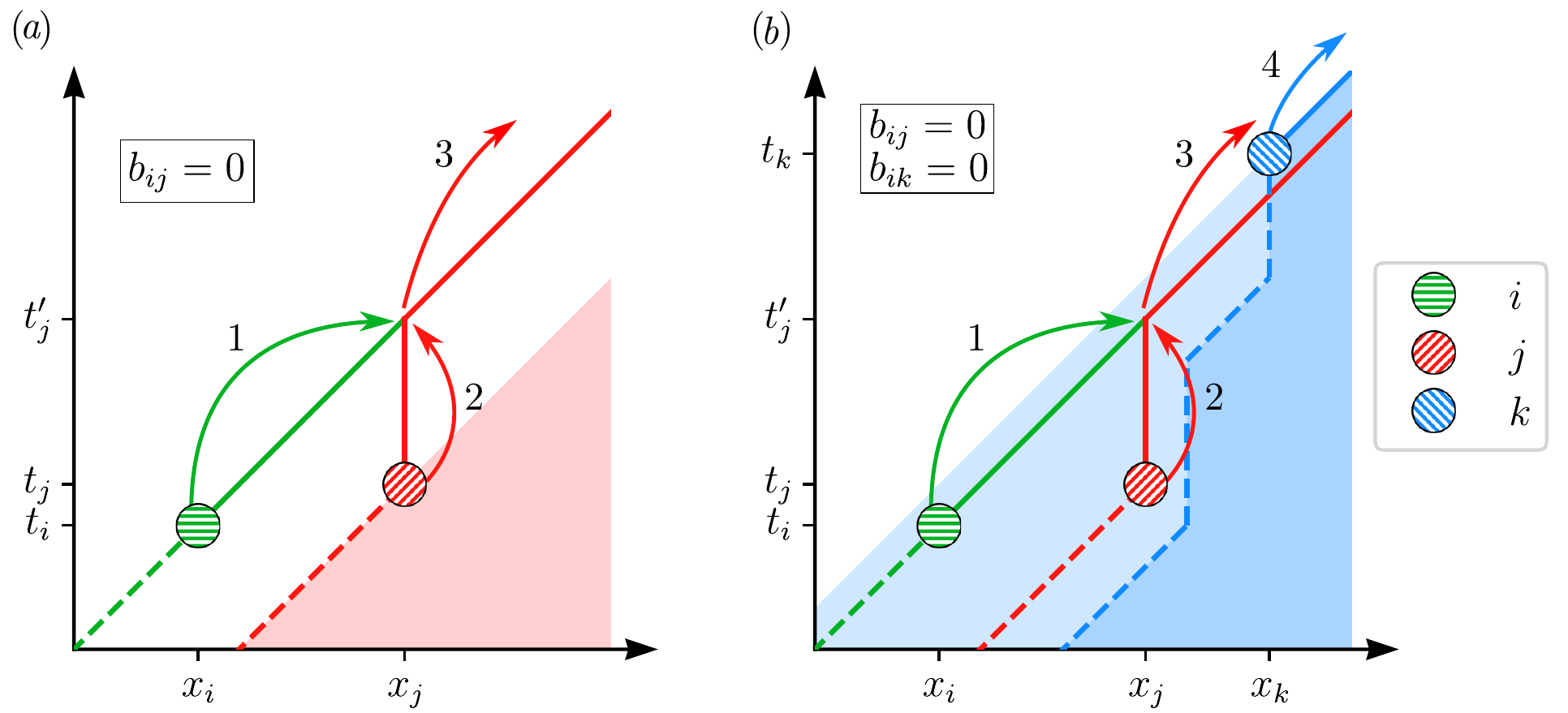}
{Horizon condition and non-crossing condition in local-time algorithms. 
\subcap{a} Sphere $i$ is above the horizon of sphere $j$ (shaded area), so that 
at the lifting of $i$ with $j$, the local-time update of $t_j$ is positive. 
\subcap{b} Sphere $i$ does not lift with $k$. It does not cross the past 
trajectory of $k$, although it violates the horizon condition with $k$ (light 
 shading) (supposing $b_{ik} = 0$). The lifting of $i$ with $j$ could in 
principle be allowed under the non-crossing condition with $k$ (dark blue 
shading), supposing $\tau_{jk} = \infty$.}
{fig:HorizonConditionForward}

\begin{algo}[Single-threaded continuous process with local times] 
With input as in \algg{alg:GlobalContinuous},  active chain $\thread=1,2, \dots$ 
is initialized (sequentially) with an active sphere $i$,  sampled from 
$\tilde{\ACAL} = \SET{i \in \ACAL, t_i \ne h'} $, and for a local-time interval 
$\tau^{\max}_\thread = \minb{\ran, h' - t_i}$, where $\ran$ is a positive random 
number. In active chain $\thread$, the active sphere $i$ moves with velocity 
$\vvactive$ for $\tau \in [0, \tau^{\max}_\thread]$ and $\diff t_i / \diff \tau 
= 1$, if the horizon condition of \eq{equ:horizon_condition} is satisfied for 
all spheres $j$. (In case of a horizon violation, the algorithm aborts.) If a 
lifting  \lifting{i}{j}{\xvec}{\xvec'}{t_i} concerns an active sphere $j$, the 
active chain $\thread$ stops with $i$ at $\xvec$. Otherwise, the local time of 
sphere $j$ is updated as $t_j(t_i^+) = t_i$ and $\ACAL =  \ACAL \setminus 
\SET{i} \cup \SET{j}$, with the active chain $\thread$ now moving $j$. The 
algorithm terminates if $\tilde{\ACAL} = \emptyset$, that is, if all the active 
spheres are stalled. Output is as for \algg{alg:GlobalContinuous}.
\label{alg:LocalContinuous}
\end{algo}
\begin{remark}[Stalled spheres]
\quot{Stalled} spheres (active spheres $i$ with $t_i = h'$) make up the set 
$\ACAL \setminus \tilde{\ACAL}$. Considering stalled spheres separately 
simplifies the sampling of $\tilde{\ACAL}$, and the restart from $h'$ for the 
next leg of the ECMC run.
\end{remark}

\begin{lemma}
If \algg{alg:LocalContinuous} terminates without a horizon violation, 
its output is  identical to that of \algg{alg:GlobalContinuous}. 
\label{lem:LocalContinuous}
\end{lemma}
\begin{proof}
We consider the final lifted configuration \lifted{\CCAL_{h'}}{ \ACAL_{h'}} of a 
run that has terminated without a horizon violation, and that has preserved a 
log of all local-time updates. The termination condition is $\tilde{\ACAL} = 
\emptyset$, so that all active spheres are stalled, with local time $h'$. We 
further consider the final lifting $l_{h''}$ in $\LCAL_{h'}$ (so that $t < h''\ 
\forall\ l_t \in \LCAL_{h''}$). Local times  of static spheres satisfy $t_i \le 
h''\ \forall i \not \in \ACAL_{h'}$. When backtracking, using 
\algg{alg:GlobalContinuous} with $- \vvactive$, from $h'$ to $h''^+$, no lifting 
takes place among active spheres (see \subfig{fig:HorizonCondition2}{a}). The 
area swept out by the active spheres cannot overlap with a static sphere $j$   
because it must have $t_j < h''$ (local times are smaller than the last lifting) 
and, on the other hand, $t_j > h''$, because of \eq{equ:horizon_condition} (see 
\subfig{fig:HorizonCondition2}{b}). The lifting 
\lifting{i}{j}{\xvec}{\xvec'}{h''} can now be undone. (From $j \in \ACAL_{h'}$ 
and $i \not \in \ACAL_{h'}$, we obtain $\ACAL_{h''} = \ACAL_{h'} \setminus 
\SET{j} \cup \SET{i}$. The updated local time $t_j(h''^-)$ can be reconstructed 
from the log. It is smaller than $h''$. The lifting is then itself eliminated: 
$\LCAL_{h''} =  \LCAL_{h'}  \setminus \SET{l_{h''}}$.) All active spheres at 
$h''^-$ now have local time $h''$. Similary, all liftings can be undone, 
effectively running \algg{alg:GlobalECMC} with $-\vvactive$ from $h'$ to $h$. As 
\algg{alg:GlobalContinuous} is time-inversion invariant, the local times at its 
liftings are the same as those  of \algg{alg:LocalContinuous}.
\end{proof}

\draftfigure[9cm]{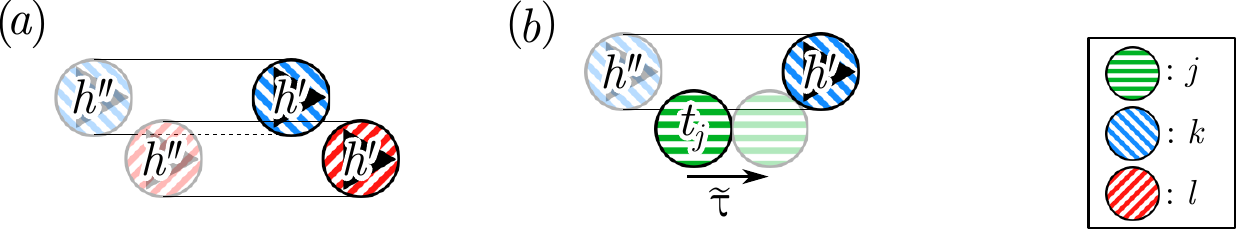}
{Backtrack using \algg{alg:GlobalContinuous} from the final configuration of 
\algg{alg:LocalContinuous}. \subcap{a} Active spheres $k$ and $l$ do not 
lift among each other. \subcap{b} A static sphere $j$ crossing the trajectory 
of active sphere $k$. This crossing is impossible because of the horizon 
condition ($\tautilde < h'- h''$ leading to $t_j > h''$, in contradiction with 
the condition $t_j < h''$).
}{fig:HorizonCondition2}

For concreteness, in the following event-driven formulation of 
\alg{alg:LocalContinuous}, the local-time interval $\tau_\thread^{\max}$ of 
an active chain $\thread$ is chosen equal to the time of flight towards the 
next lifting. 

\begin{algo}[Single-threaded ECMC with local times]
With input as in \algg{alg:GlobalContinuous}, for each (sequential) active chain 
$\thread=1,2,\dots$, an active sphere $i$  is sampled from $\tilde{\ACAL} = 
\SET{i \in \ACAL, t_i \ne h'}$. The horizon conditions of 
\eq{equ:horizon_condition} are checked for all\footnote{at most three spheres 
$j$ can have finite $\tau_{ij}$ for any $i$, see \sect{sec:ConstraintGraph}} 
spheres $j$ that can have finite time of flight $\tau_{ij}$. The algorithm 
aborts if a violation occurs. Otherwise, $i$ is moved forward to $\minb{t_i + 
\tau_{ij}, h'}$, and the local time of $i$ and $j$ are updated to that time. The 
active chain stops if $j$ is an active sphere or if the local time equals $h'$. 
Otherwise, the move corresponds to a lifting \lifting{i}{j}{\xvec}{\xvec'}{t_i + 
\tau_{ij}} and $\ACAL = \ACAL \setminus \SET{i} \cup \SET{j}$, with the active 
chain now moving $j$. The algorithm terminates if $\tilde{\ACAL} = \emptyset$. 
Output is as for \algg{alg:GlobalContinuous}.
\label{alg:LocalECMC}
\end{algo}
\begin{availability}
\algg{alg:LocalECMC} is implemented in 
\verb#SingleThreadLocalTimeECMC.py# and tested in the
\verb#PValidateECMC.sh# script.
\end{availability}

\begin{remark}[Partial validation]
In \sect{sec:Validation}, a variant of \algg{alg:LocalECMC} is used to validate 
part of a run, even if it does not terminate correctly. When a sphere $i$ 
detects a horizon violation, its time $t_i$ is recorded. At $h'$, the 
set 
$\LCAL_{t^*}$ of all liftings up to the earliest horizon violation, at $t^*$, 
agrees with the corresponding partial list of liftings for 
\algg{alg:GlobalContinuous}.
\label{rem:PartialValidation}
\end{remark}

\subsection{Multithreaded ECMC (sequential-consistency model)}
\label{sec:MultiSeqConsist}

\alg{alg:MultiSeqConsist}, the subject of the present section, is a model 
shared-memory ECMC on $k$ threads, that is, on as many threads as there are 
active spheres. The algorithm adopts the sequential-consistency 
model~\cite{Lamport1979}. We rigorously prove its correctness for small 
test suites by mapping the multithreading stage of this algorithm to an 
absorbing 
Markov chain.  The algorithm allows us to show that certain seemingly innocuous 
modifications of \alg{alg:MultiCPP} (the C++ implementation) contain bugs that 
are too rare to be detected by routine testing. 
\draftfigure[\columnwidth]{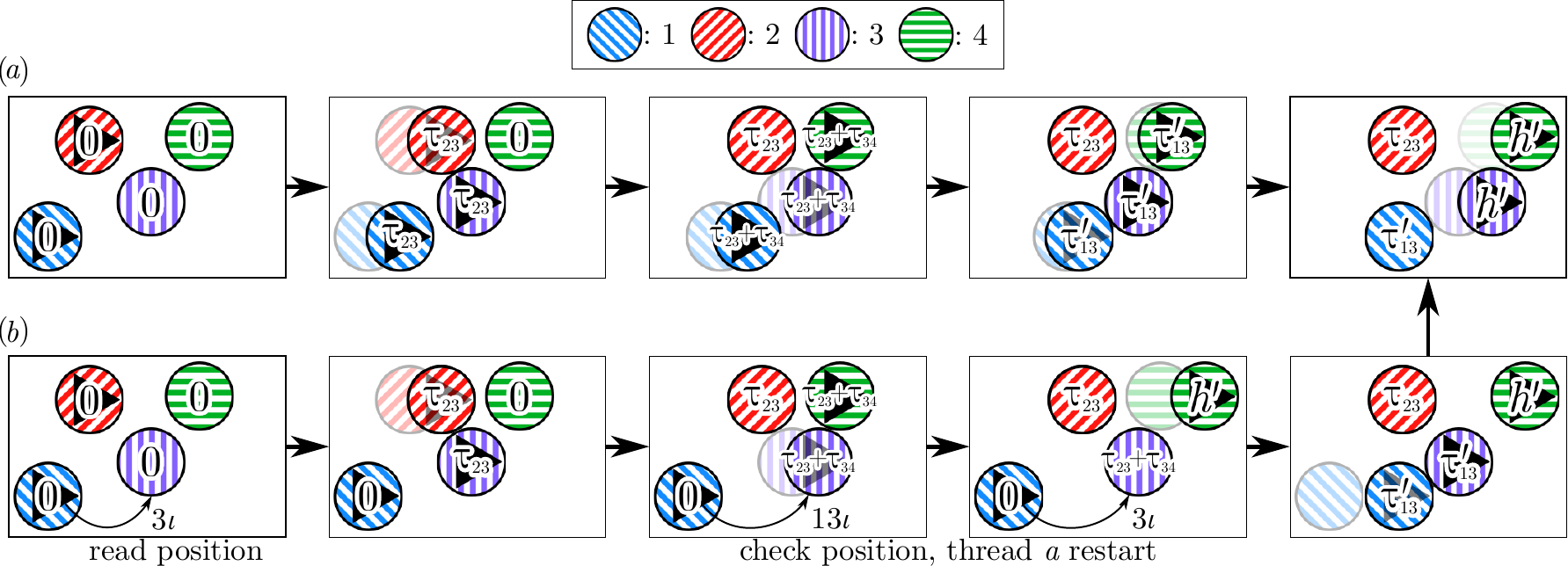} 
{\algg{alg:MultiSeqConsist}, as applied to the \texttt{SequentialC4} test 
suite. 
\subcap{a} Reference set \Lref from \algg{alg:GlobalECMC}. 
\subcap{b} Run of \algg{alg:MultiSeqConsist} involving a \quot{lock-less} lock 
rejection on the position of sphere $3$.} {fig:LocalGlobalComparison} 

The algorithm has three stages. In the (sequential) initialization
stage, 
it inputs a lifted initial configuration and maps 
each active sphere to a thread. This is followed by the multithreading stage, 
where each active chain progresses independently,  checking the horizon 
conditions in its local environment. The algorithm concludes with the 
(sequential) output stage. 

At each step of the multithreading stage of \alg{alg:MultiSeqConsist}, a switch 
randomly selects one of the $k$ statements (one for each thread $a,b, \dots$) 
contained in a buffer as $\buffer{\text{next}_a,\text{next}_b, \dots}$. The 
selected statement is executed on the corresponding thread, and then the buffer 
is updated. The random sequence of statements mimics the absence of 
thread synchronization except at breakpoints. All threads possess an 
absorbing {\bf wait} statement. When it is reached throughout, the algorithm 
progresses to the output stage, followed by successful termination. 
The program aborts when a thread detects a horizon violation. For our test 
suites, we prove by explicit construction that each state is connected to at 
least one of the absorbing states, but we lack a general proof of validity for 
arbitrary configurations and general $N$.

In \alg{alg:MultiSeqConsist}, each sphere has three attributes, namely a tag, a 
local time and a position. The sphere's tag indicates whether it is active on a 
thread $\thread$, stalled, or static. All threads have read/write 
access to the 
attributes of all spheres. A state of the Markov chain is constituted by the 
spheres with their attributes, some local variables and the buffer content.

\begin{algo}[Multithreaded ECMC (sequential-consistency model)]
At breakpoint $h=0$, a lifted initial configuration \lifted{\CCAL_h}{\ACAL_h} is 
input (see \fig{fig:LocalGlobalComparison} for the example with four spheres). 
All local times are set to $h=0$, all tags are put to static, except for the 
active spheres, whose tags correspond to their thread $\thread$.
The buffer is set to  \buffer{1_1,\ 1_\thread \TO 1_k}. A random switch 
selects one buffer element. The corresponding statement is executed on its 
thread, and the buffer is replenished. The following provides pseudo-code for 
the multithreading stage ($i_\thread$ is the active sphere,
$j_\thread$ the target sphere, and $distance_\thread$ the difference between  
$h'$ and the local time, all on thread $\thread$.): \\
\noindent
\begin{center}
\begin{tabular}{rl}
$1\thread$&$\tau_\thread \leftarrow \distance_\thread$; $j_\thread \leftarrow 
i_\thread$; $x_\thread \leftarrow \infty$\\
$2\thread$&{\bf for} $\nbr$ {\bf in} ${\text \{1,2,...,n\}}\setminus i_\thread$ 
{\bf 
:}\\
$3\thread$&\hspace{3mm}$x_\nbr \leftarrow \nbr.x$\\
$4\thread$&\hspace{3mm}$\tau_{i\nbr} \leftarrow x_\nbr - i_\thread.x - 
b_{i\nbr}$\\
$5\thread$&\hspace{3mm}{\bf if} $i_\thread.t\! + \!\tau_{i\nbr}\! < 
\!\nbr.t${\bf\ :} 
{\bf 
abort}\\ 
$6\thread$ &\hspace{3mm}{\bf if} $\tau_{i\nbr} < \tau_\thread${\bf\ :} \\ 
$7\thread$&\hspace{6mm}$j_\thread\leftarrow \nbr$ \\ 
$8\thread$&\hspace{6mm}$x_\thread\leftarrow x_{\nbr}$ \\
$9\thread$&\hspace{6mm}$\tau_\thread \leftarrow \tau_{i\nbr}$ \\
$10\thread$&$j_\thread.tag.$CAS$(\static,\thread)$\\
$11\thread$ &{\bf if} $j_\thread.tag = \thread${\bf\ :}\\
$12\thread$&\hspace{3mm}{\bf if} $\tau_\thread < \distance_\thread${\bf\ :}\\
$13\thread$&\hspace{6mm}{\bf if} $x_\thread = j_\thread.x${\bf\ :}\\
$14\thread$&\hspace{9mm}$j_\thread.t \leftarrow  i_\thread.t+\tau_\thread$\\
$15\thread$&\hspace{9mm}$i_\thread.t \leftarrow  i_\thread.t+\tau_\thread$\\
$16\thread$&\hspace{9mm}$i_\thread.x \leftarrow i_\thread.x + \tau_\thread$\\
$17\thread$&\hspace{9mm}$i_\thread.tag \leftarrow \static$\\
$18\thread$&\hspace{9mm}$\distance_\thread \leftarrow 
\distance_\thread-\tau_\thread$\\
$19\thread$&\hspace{9mm}$i_\thread \leftarrow j_\thread$\\
     &\hspace{6mm}{\bf else\ :}\\
$20\thread$&\hspace{9mm}$j_\thread.tag \leftarrow \static$ \\
     &\hspace{9mm}{\bf goto} $1$\\
&\hspace{3mm}{\bf else\ :}\\
$21\thread$&\hspace{6mm}$i_\thread.t \leftarrow i_\thread.t+\tau_\thread$\\
$22\thread$&\hspace{6mm}$i_\thread.x \leftarrow i_\thread.x + \tau_\thread$\\
$23\thread$&\hspace{6mm}$\distance_\thread \leftarrow 0$\\
$24\thread$&\hspace{6mm}$i_\thread.tag \leftarrow \stalled$\\
     &{\bf else\ :} {\bf goto} $1$\\
$25\thread$&{\bf if} $\distance_\thread > 0${\bf\ :} {\bf goto} $1$\\
$26\thread$&{\bf wait}\\  
\end{tabular} 
\end{center}
When all $k$ threads have reached their {\bf wait} statements, the algorithm 
proceeds to its output stage. Output is as for 
\algg{alg:GlobalContinuous}.
\label{alg:MultiSeqConsist}
\end{algo}
\begin{availability}
\verb#SequentialMultiThreadECMC.py#. implements 
\algg{alg:MultiSeqConsist}. It also constructs all states connected to the 
initial state and traces them  to the absorbing states. \end{availability}

\begin{remark}[Illustration of pseudocode]
The multithreading stage of \algg{alg:MultiSeqConsist} corresponds in $k$ 
independent programs running independently. In the sequential-consistency model, 
the space of programming statements is thus $k$-dimensional (one sequence 
$(1\thread \TO 26\thread)$ per thread), and each displacement in this space 
proceeds along a randomly chosen coordinate axis. As an example, if for a buffer 
\buffer{next_1 \TO next_\thread=20\thread \TO next_k} the switch selects thread 
$\thread$, then the tag of target particle $j_\thread$ is set to \quot{static}, 
and the buffer is updated to \buffer{next_1 \TO next_\thread=1\thread \TO 
next_k}. The thread $\thread$ will thus be restarted at its next selection.
\end{remark}

\noindent
The compare-and-swap (CAS) statement in $10\thread$ of  
\alg{alg:MultiSeqConsist}  amounts to a single-line {\bf if}. It is equivalent 
to: \quot{{\bf if} $j_\thread.tag = \static$ {\bf\ :} $j_\thread.tag = \thread$} 
(if $j$ is static, then it is set to active on thread $\thread$ (see 
\rem{rem:CASMeaning} for a discussion).

We prove correctness of \alg{alg:MultiSeqConsist}, for the \verb#SequentialC4# 
test suite with $N=$ and $k=2$ (see \fig{fig:LocalGlobalComparison}), that we 
later extend to the \verb#SequentialC5# test suite with $N=5$.

\begin{lemma}
\label{lem:MultiSeqConsist}
If \algg{alg:MultiSeqConsist} terminates without a horizon violation, its output 
(for the \verb#SequentialC4# test suite) is identical to that of 
\algg{alg:GlobalContinuous}.
\end{lemma}

\begin{proof}
In the \verb#SequentialC4# test suite with 
threads  \quot{$\threada$} and \quot{$\threadb$} we suppose that the switch 
samples $\threada$ and  $\threadb$ with equal probabilities.
The two-thread stage of 
\algg{alg:MultiSeqConsist} then consists in a finite Markov chain with the 
$3670$ states $S_n$ that are accessible from the initial state. The {\bf abort} 
state has no buffer content. All other $3669$ states comprise  the buffer 
\buffer{\text{next}_{\threada},\text{next}_{\threadb}}, the sphere objects (the 
spheres and their attributes: tag, local time, position), and some 
thread-specific local variables. One iteration of the Markov chain (selection of 
$\text{next}_\threada$ or $\text{next}_\threadb$, execution of the corresponding 
statement, buffer update) realizes the transition from $S_n$ to a state $S_m$ 
with probability $T_{nm} = 1/2$. The $3670\times 3670$ transition matrix $T = 
(T_{nm})$ has unit diagonal elements for the {\bf abort}, and for the unique 
{\bf terminate} state with buffer \buffer{26\threada,26\threadb}, which are both 
absorbing states of the Markov chain. Furthermore, we can show explicitly that 
all $3670$ states have a finite probability to reach an absorbing state in a 
finite number of steps. This proves that the  Markov chain is absorbing. For an 
absorbing Markov chain, all states that are not absorbing are transient, and 
they die out at large times. The algorithm thus either ends up in the unique 
{\bf terminate} state that corresponds to successful completion, or else in the 
{\bf abort} state. 
\end{proof}

\noindent
All states of the Markov chain may be projected onto their buffer 
\buffer{\text{next}_{\threada},\text{next}_{\threadb}} and visualized (see 
\fig{fig:Coincidence_table})).

\draftfigure[8cm]{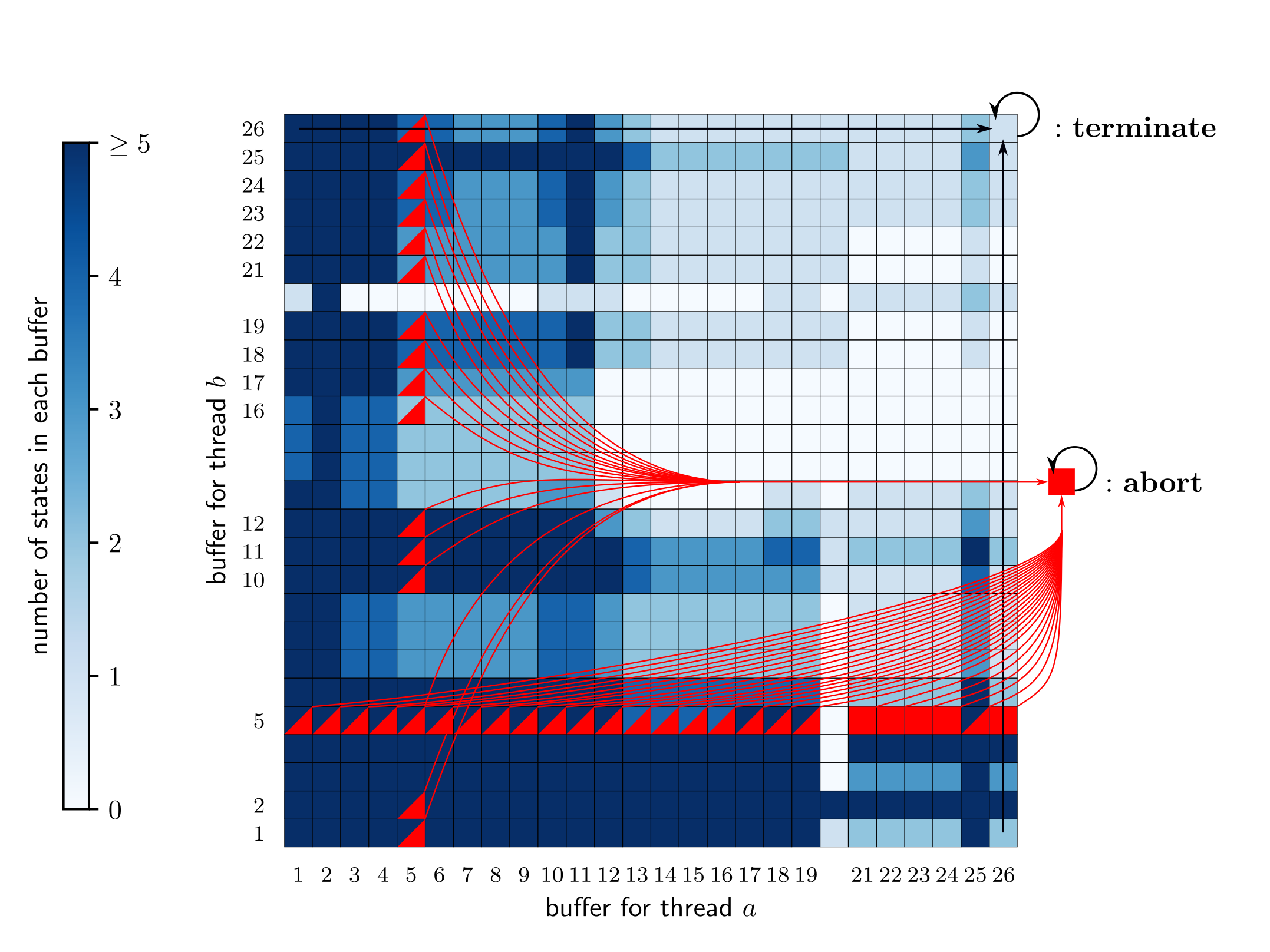}{The $3670$ states in 
\algg{alg:MultiSeqConsist} for the \texttt{SequentialC4} test suite projected 
onto the buffer content \buffer{\text{next}_{\threada},  \text{next}_{\threadb}} 
(see \fig{fig:LocalGlobalComparison}). The {\bf terminate} buffer 
\buffer{26\threada, 26\threadb} corresponds to a single state. 
}{fig:Coincidence_table} 

\begin{remark}[CAS statement]
The CAS statements (see $10\thread$ in \algg{alg:MultiSeqConsist}) acquire 
their full meaning
in the multithreaded \algg{alg:MultiCPP} The way in which they differ 
from simple {\bf if} statements can already be illustrated in the simplified 
setting. We suppose two threads $\threada$ and $\threadb$. Then, with 
$\target_\thread$ the target sphere on thread $\thread$, a buffer content 
\buffer{10\threada, 10\threadb}: \\
\begin{minipage}[t]{.5\linewidth}
 \vspace{0pt}
\begin{tabular}{rl}
\vdots & \vdots \\
$10\threada$&$j_{\threada}.tag.$CAS($\static,\threada$)\\
$11\threada$&{\bf if} $j_\threada.tag = \threada${\bf\ :}\\
\vdots & \vdots \\
\end{tabular}
\end{minipage}
\vline
\begin{minipage}[t]{.5\linewidth}
 \vspace{0pt}
\begin{tabular}{rl}
\vdots & \vdots \\
$10\threadb$&$j_{\threadb}.tag.$CAS$(\static,\threadb)$\\
$11\threadb$&{\bf if} $j_\threadb.tag = \threadb${\bf\ :}\\
\vdots & \vdots \\
\end{tabular}
\end{minipage}
can belong to a state with $j_{\threada}= 3 = j_{\threadb}=3$.\footnote{This 
corresponds to the lifted configuration of 
\subfig{fig:LocalGlobalComparison}{h}}. If the statement $10a$ is selected, 
sphere $3$ becomes active on thread $\threada$ (through the statement $3.tag= 
\threada$). In contrast, if the switch selects $10b$, sphere $3$ becomes active 
on thread $\threadb$. The program continues consistently for both switch 
choices, because the selection is made in a single (\quot{atomic}) step on each 
thread and because the sequential-consistency model avoids conflicting memory 
assignments. In contrast, if the switch selection from \buffer{10a,10b} is split 
as: \\
\begin{minipage}[t]{.5\linewidth}
 \vspace{0pt}
\begin{tabular}{rl}
\vdots & \vdots \\
$10'\threada$&{\bf if} $j_{\threada}.tag = \static${\bf\ :}\\
$10''\threada$&\hspace{3mm}$j_{\threada}.tag = \threada$\\
$11\threada$&{\bf if} $j_\threada.tag = \threada${\bf\ :}\\
\vdots & \vdots \\
\end{tabular}
\end{minipage}
\vline
\begin{minipage}[t]{.5\linewidth}
 \vspace{0pt}
\begin{tabular}{rl}
\vdots & \vdots \\
$10'\threadb$&{\bf if} $j_{\threadb}.tag = \static${\bf\ :}\\
$10''\threadb$&\hspace{3mm}$j_{\threadb}.tag = \threadb$\\
$11\threadb$&{\bf if} $j_\threadb.tag = \threadb${\bf\ :}\\
\vdots & \vdots \\
\end{tabular}
\end{minipage}
the sequence $10'\threada \to 10'\threadb \to 10''\threada \to 11\threada \to 
10''\threadb \to 11\threadb$ results in sphere $3$ first becoming active on 
thread $\threada$ (and the thread continuing as if this remained the case), and 
then on thread $\threadb$, which is inconsistent. In \alg{alg:MultiCPP}, the C++ 
implementation of multithreaded ECMC, the CAS likewise keeps this selection step 
atomic, and likewise excludes memory conflicts among all threads during this 
step. It thus plays the role of a lightweight memory lock. 
\label{rem:CASMeaning}
\end{remark}

\noindent
\alg{alg:MultiSeqConsist} features lock-free 
programming, which is also a key ingredient of
\alg{alg:MultiCPP}.
\begin{remark}[Lock-free programming]
To illustrate lock-free programming in \algg{alg:MultiSeqConsist}, we 
consider two threads, $\threada$ and $\threadb$. \\
\noindent
\begin{minipage}[t]{.5\linewidth}
 \vspace{0pt}
\begin{tabular}{rl}
$7\threada$&\hspace{6mm}$j_\threada\leftarrow \nbr$ \\
$8\threada$&\hspace{6mm}$x_\threada\leftarrow x_{\nbr}$ \\
\vdots & \vdots \\
$10\threada$&$j_\threada.tag.$CAS$(\static,\threada)$\\
\vdots & \vdots \\
$13\threada$ &\hspace{6mm}{\bf if} $x_\threada = j_\threada.x${\bf\ :}  \\
\vdots & \vdots \\
\end{tabular}
\end{minipage}
\vline
\begin{minipage}[t]{.5\linewidth}
 \vspace{0pt}
\begin{tabular}{rl}
\vdots & \vdots \\
$16\threadb$&\hspace{9mm}$i_\threadb.x \leftarrow i_\threadb.x + 
\tau_\threadb$\\
$17\threadb$&\hspace{9mm}$i_\threadb.tag \leftarrow \static$\\
\vdots & \vdots \\
\end{tabular}
\end{minipage}
\noindent
The identification of the target sphere $j_{\threada}$ on thread $\threada$ 
(statements $7\threada$ and $8\threada$) would be compromised if, before locking 
through the CAS statement at $10\threada$, it was changed in thread $\threadb$, 
where the same sphere $i_{\threadb}$ is active (see statements $16\threadb, 17 
\threadb$). However, the statement $13\threada$ checks that sphere 
$j_{\threada}$ has not moved. If this condition is not satisfied, the thread 
$\threada$ will end up being restarted (through statement $20\threada$). (See 
also \rem{rem:MemoryOrder}.)
\label{rem:StatementOrdering}
\end{remark}
\noindent

\subsection{Multithreaded ECMC (C++, OpenMP implementation)}
\label{sec:MultiCPP}
\alg{alg:MultiCPP}, discussed in this section, translates 
\alg{alg:MultiSeqConsist} into C++ (OpenMP). The CAS statement and lock-free 
programming assure its efficiency. A sphere's attributes are again its position, 
its local time, and its tag. The latter is a an atomic variable. 
We refer to line numbers in \alg{alg:MultiSeqConsist}.

\begin{algo}[Multithreaded ECMC (C++, OpenMP)]
With initial values as in \algg{alg:GlobalContinuous}, thread management is 
handled by OpenMP. The number of threads can be smaller than the number of 
active spheres. The multithreading stage transliterates the one of 
\algg{alg:MultiSeqConsist}. Statement $2\thread$ of \algg{alg:MultiSeqConsist} 
is implemented through a constraint graph (see \sect{sec:ConstraintGraph}). 
Statements $10\thread$ through $13 \thread$ are expressed as follows in 
\verb#MultiThreadECMC.cc#:\\
\begin{tabular}{rl}
$10\thread$  $\rightarrow$  & 
\verb#j->tag.compare_exchange_strong(...static,...#\\
$11\thread$  $\rightarrow$  & 
\verb#if (j->tag.load(memory_order) == iota)#  \\
$12\thread$  $\rightarrow$  & 
\verb#if (tau < distance) #\\
$13\thread$  $\rightarrow$  & 
\verb#if (x == j->x)#, \\
\end{tabular} \\
where the \verb#memory_order# qualifier may take on different values (see 
\sect{sec:Validation}). Important differences with \algg{alg:MultiSeqConsist} 
are discussed in \remtwo{rem:Necklace}{rem:MemoryOrder}. Output is as for 
\algg{alg:GlobalContinuous}.

\label{alg:MultiCPP}
\end{algo}
\begin{availability}
\algg{alg:MultiCPP} is implemented in \verb#MultiThreadECMC.cc#. It is executed 
in several validation and benchmarking scripts (see \sect{sec:Validation}).
\end{availability}

\begin{remark}[Active-sphere necklaces]
\algg{alg:MultiSeqConsist} restarts thread $\thread$ if the target sphere $j$ 
(for an active sphere $i$ on the thread) is itself active on another thread. 
With periodic boundary conditions, active-sphere necklaces, where all target 
spheres are active, can  deadlock the algorithm. To avoid this, 
\algg{alg:MultiCPP} moves sphere $i$ up to contact with $j$ before restarting 
(this is also used in \algg{alg:LocalECMC}).
\label{rem:Necklace}
\end{remark}

The source code of \alg{alg:MultiCPP} essentially translates that of 
\alg{alg:MultiSeqConsist}. The compiler may however change  the order of 
execution for some statements in order to gain efficiency. (The memory access in 
modern multi-core processors can also be very complex and, in particular, 
thread-dependent.) Attributes, such as the \verb#memory_order# qualifier in the 
CAS statement, may constrain the allowed changes of order. The reordering 
directives adopted in \alg{alg:MultiCPP} were chosen and validated with the help 
of extensive runs from randomly generated configurations. However, subtle 
pitfalls escaping notice through such testing  can be exposed by  explicitly 
reordering statements in \alg{alg:MultiSeqConsist}.

\begin{remark}[Memory-order directives in 
\alggtwo{alg:MultiSeqConsist}{alg:MultiCPP}]
In the \verb#SequentialC5# test suite  with $N=5$, interchanging statements 
$15\thread$ and $16\thread$ in \algg{alg:MultiSeqConsist} yields a spurious 
absorbing state, and invalidates the algorithm. The same  test suite can also be 
input into \algg{alg:MultiCPP}, where it passes the \verb#Ordering.sh# 
validation 
test, even if the statements in \verb#MultiCPP.cc# corresponding to 
$15\thread$ and $16\thread$ are exchanged. However, a $1\mu$s \verb#pause#  
statement introduced in the C++ program between what corresponds 
to the (interchanged) statements $15\thread$ and $16\thread$
produces a 
$sim 1\%$ error 
rate, illustrating that \algg{alg:MultiCPP} is unsafe without a protection of 
the order of the said statements. Safety may be increased through
atomic
position and local time variables, allowing the use of the 
\verb#fetch_add()# operation to displace spheres.
\label{rem:MemoryOrder}
\end{remark}
\noindent

\section{Tools, validation protocols, benchmarks, and extensions}
\label{sec:Implementation}
We now discuss the implementations of the algorithms of \sect{sec:Algorithms}, 
as well as their validation protocols, benchmarks, and possible extensions. We 
also discuss the prospects of this method beyond this paper's focus on 
the interval between two breakpoints $h$ and $h'$.

In our implementation of \algthree{alg:GlobalECMC}{alg:LocalECMC}{alg:MultiCPP}, 
a directed constraint graph encodes the possible pairs of active and target 
spheres as arrows \arrow{i}{j} (see \sect{sec:ConstraintGraph}). The outdegree 
of this graph is at most three, and a rough constraint graph \Gthree with, 
usually,
outdegree three for all vertices is easily generated. 
\Gthree may contain redundant arrows that cannot correspond to liftings.  Our 
pruning algorithm eliminates many of them. 
We also prove that \Gmin, the minimal constraint graph, is planar. This 
may 
be of importance if disjoint parts of the constraint graph are stored on 
different CPUs, each with a number of dedicated threads. In general, we expect 
constraint graphs to be a useful tool for hard-sphere production codes, with 
typically \bigOb{N} liftings between changes of $\vvactive$.

Validation scripts are discussed in \sect{sec:Validation}. Scripts check that  
the  liftings of standard cell-based ECMC are all accounted for  
in the used constraint graph. For the ECMC algorithms of 
\sect{sec:Algorithms}, the 
set $\LCAL$ of liftings provides the complete history of each run, and 
scripts check that they corresponds to \Lref.

In \sect{sec:Benchmarks}, we then benchmark 
\alg{alg:MultiCPP} and demonstrate a speed-up  by an order of 
magnitude for a single CPU with 40 threads on an \xCPU (see 
\sect{sec:Benchmarks}). The overhead introduced by multithreading ($\sim 2.4$) 
is very reasonable. We then discuss possible extension of our methods (see 
\sect{sec:Extensions}).

\subsection{Constraint graphs}
\label{sec:ConstraintGraph}

For a given initial condition $\CCAL_h$ and velocity $\vvactive$, arrows 
\arrow{i}{j} of the constraint graph $\GCAL_{\vvactive}$ represent possible 
liftings \lifting{i}{j}{\xvec}{\xvec'}{t}~\cite{KapferPolytope2013}. Arrows 
remain unchanged between breakpoints because spheres $i$ and $j$ with a 
perpendicular  distance of  less than $2 \sigma$ cannot hop over one another 
(this argument can be adapted to periodic boundary conditions), and pairs with 
larger perpendicular distance are absent from $\GCAL$. All constraint graphs 
$\GCAL_{\vvactive}$ are supersets of a minimal constraint graph 
$\Gmin[\vvactive]  \equiv \Gmin[-\vvactive]$ (where the equivalence is 
understood as $\arrow{i}{j}_{\vvactive} \equiv \arrow{j}{i}_{-\vvactive} $).

\begin{remark}[Constraint graphs and convex polytopes]
Each arrow \arrow{i}{j} of the  constraint graph $\GCAL_{\vvactive}$ provides 
(for $\vvactive = (1,0)$) an inequality
\begin{equation}
\begin{aligned}
x_i & \le x_j - b_{ij}
\end{aligned}
\label{equ:Inequality}
\end{equation}
that is tight ($x_i = x_j - b_{ij}$) when $i$ lifts to $j$ at contact (if there 
are configurations $\CCAL$ where it is tight, then \arrow{i}{j}  belongs to 
\Gmin[\vvactive]). The set of inequalities defines a convex polytope. With 
periodic boundary conditions (unaccounted for in \eq{equ:Inequality}),
this polytope is infinite in the direction corresponding to uniform translation 
of all spheres with $\vvactive$ (see~\cite{KapferPolytope2013}).
\end{remark}

\begin{remark}[Constraint graphs and irreducibility]
Rigorously, we define the constraint 
graph \Gmin[\vvactive] as the set of arrows \arrow{i}{j} that 
are 
encountered from $\CCAL_h$ by \algg{alg:GlobalContinuous} (or, equivalently, 
\algg{alg:GlobalECMC}) at liftings $\lift{t}\ \forall t \in (-\infty, \infty)$. 
The liftings for $t < 0$ can be constructed because of time-reversal invariance 
(see \rem{rem:TimeReversalInvariance}). For the same reason, we have 
$\Gmin[\vvactive] \equiv \Gmin[-\vvactive]$, and the set of arrows reached from 
$\CCAL_h$ is equivalent to that reached from any configuration that is reached 
from $\CCAL$ (and in particular $\CCAL_{h'}$). While we expect ECMC to 
be 
irreducible in the polytope defined through the inequalities in 
\eq{equ:Inequality}, we do not require irreducibility for the definition of 
\Gmin.
\end{remark}

Between breakpoints, the active sphere $i$ can lift to at most three other 
spheres, namely the sphere $j^0$ minimizing the time of flight $\tau_{ij}$ in a 
corridor of width $2\sigma$ around the center of $i$ , and likewise the 
closest-by sphere $j^+$ in the corridors  $[\sigma, 2 \sigma]$ and sphere $j^-$ 
in the corridor $[-2\sigma, -\sigma]$ (see \subfig{fig:ThreeLane}{a}). The set 
of arrows $\SET{\arrow{i}{j^0}, \arrow{i}{j^-}, \arrow{i}{j^+}\ \forall i \in 
\SET{1 \TO N}}$ constitutes the constraint graph \Gthree[\vvactive], which is 
thus easily computed. Except for small systems (when the corridors may be 
empty) \Gthree has outdegree three  for all spheres $i$. However, its indegree 
is not fixed. The constraint graph \Gthree is not necessarily locally 
planar,\footnote{\quot{Locally planar} means that any subgraph that does not 
sense the periodic boundary conditions is planar.} and in the embedding provided 
by the sphere centers of a given configuration, non-local arrows can be present 
(see \subfig{fig:ConstraintGraph}{a}). However, \Gmin can be proven to be 
locally planar
(see \subfig{fig:ConstraintGraph}{b})).

\draftfigure[7cm]{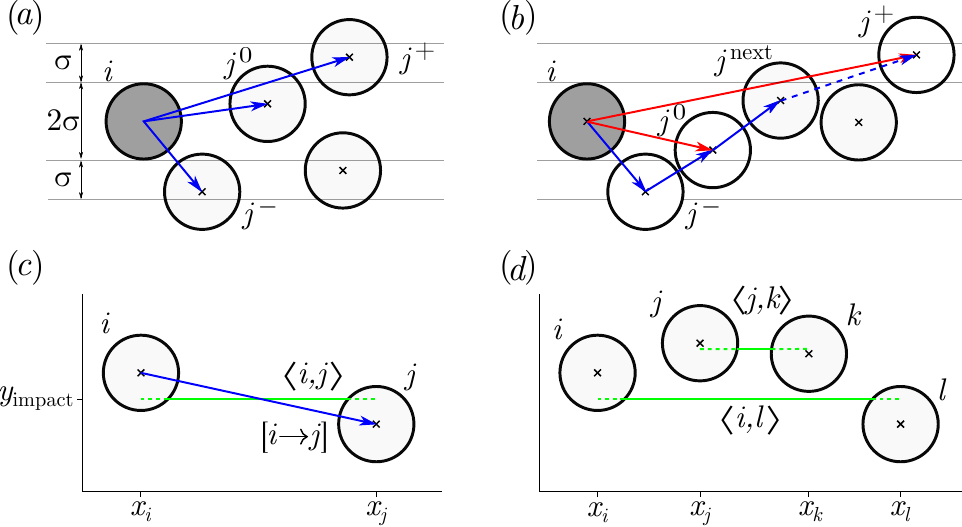}{Constraint graphs, pruning, and planarity.
\subcap{a} Corridors of an active sphere $i$, with arrows \arrow{i}{j^-}, 
\arrow{i}{j^0},
and \arrow{i}{j^+} belonging to \Gthree[(1,0)]. 
\subcap{b} Pruning of an arrow \arrow{i}{j^+} through a sphere 
$j^{\text{next}}$ without there being an arrow 
\arrow{i}{j^{\text{next}} }.
\subcap{c} Spheres $i$ and $j$, arrow \arrow{i}{j}, and impact path 
\impact{i}{j}.
Other spheres cannot cover \impact{i}{j}.
\subcap{d} Spheres $i,j,k,l$ with $x_i< \dots<x_l$ and impact paths
\impact{i}{l} and \impact{j}{k}. }
{fig:ThreeLane}
\begin{lemma}
The graph \Gmin is locally planar, and any sphere configuration that can be 
reached between breakpoints provides a locally planar embedding.
\end{lemma}
\begin{proof}
We first consider two spheres $i$ and $j$ for $\vvactive = (1,0)$ in the plane 
(without taking into account periodic boundary conditions). The arrow 
\arrow{i}{j} is drawn by connecting the centers of $i$ and $j$. The impact path 
\impact{i}{j} is the horizontal line segment connecting  
$(x_i,y^{\text{impact}})$  and $(x_j,y^{\text{impact}})$ where 
$y^{\text{impact}}$ is the vertical position at which the two spheres can touch 
by moving them with $\vvactive$ (see \subfig{fig:ThreeLane}{c}). If the arrow 
\arrow{i}{j} exists, no other sphere can intersect the impact path 
\impact{i}{j}.

For four spheres $i,j,k,l$, we now show that no two arrows between spheres can 
cross each other. The $x$-values can be ordered as $x_i < x_j < x_k < x_l$ 
(again without taking into account periodic boundary conditions). Two arrows 
between three spheres trivially cannot cross. For arrows between two pairs of 
spheres, trivially arrows \arrow{i}{j} and \arrow{k}{l} cannot cross. Likewise, 
if there is an arrow \arrow{i}{k} then sphere $j$ must be on one side of the 
impact path \impact{l}{k}, and $k$ must be on the other side of \impact{j}{l}, 
so that arrows \arrow{i}{k} and \arrow{j}{l} cannot cross. Finally, if  arrow 
\arrow{i}{l} exists, then $j$ and $k$ must be on the same side of the impact 
path \impact{i}{l} in order to have an impact path. But then, \arrow{j}{k} 
cannot cross \arrow{i}{l} (see \subfig{fig:ThreeLane}{d}).
\end{proof}

The minimal constraint graph  \Gmin is more difficult to compute than \Gthree 
because the underlying \quot{redundancy detection} problem is not strictly 
polynomial in system size, although practical algorithms 
exist~\cite{Fukuda2016}. However, \Gthree can be pruned of redundant constraints 
that correspond to pairs of spheres $i$ and $j$ that are prevented from lifting 
by other spheres. For example, given arrows \arrow{i}{j}, \arrow{j}{k} and  
\arrow{i}{k}, the latter can be \quot{first-order} pruned (eliminated with one 
intermediary, namely $j$) if $b_{ij} + b_{jk} > b_{ik}$ (in \fig{fig:ThreeLane}, 
\arrow{i}{j^+} can be pruned for this reason). The presence of the arrow 
\arrow{j}{k} is not necessary to make this argument work (see 
\fig{fig:ThreeLane}{b}). Pruning can be taken to higher orders. To second order, 
if  $b_{ij} + b_{jk} + b_{kl}> b_{il}$, then the arrow \arrow{i}{l} can be 
eliminated. Finally, any arrow \arrow{i}{j} in $\GCAL_{\vvactive}$ can be pruned 
through symmetrization if it is unmatched by \arrow{j}{i} in 
$\GCAL_{-\vvactive}$ because $\Gmin[\vvactive] \equiv \Gmin[-\vvactive]$, with 
$\GCAL[\vvactive]$ and $\GCAL[-\vvactive]$ obtained separately (see 
\rem{rem:TimeReversalInvariance}).

Rarely, arrows can be eliminated by symmetrizing 
graphs that were pruned to third or fourth order, and constraint graphs that 
are 
obtained in this way appear close to \Gmin (see \sect{sec:Validation}).

\begin{availability}
The constraint graph \Gthree is constructed in \verb#GenerateG3.py# and pruned 
to $\GCAL$ in \verb#PruneG.py#. The program \verb#GraphValidateCellECMC.cc# runs 
cell-based ECMC to verify the consistency of $\GCAL$.
\end{availability}

\begin{figure}[htb]
\begin{center}
\includegraphics[width=\linewidth]{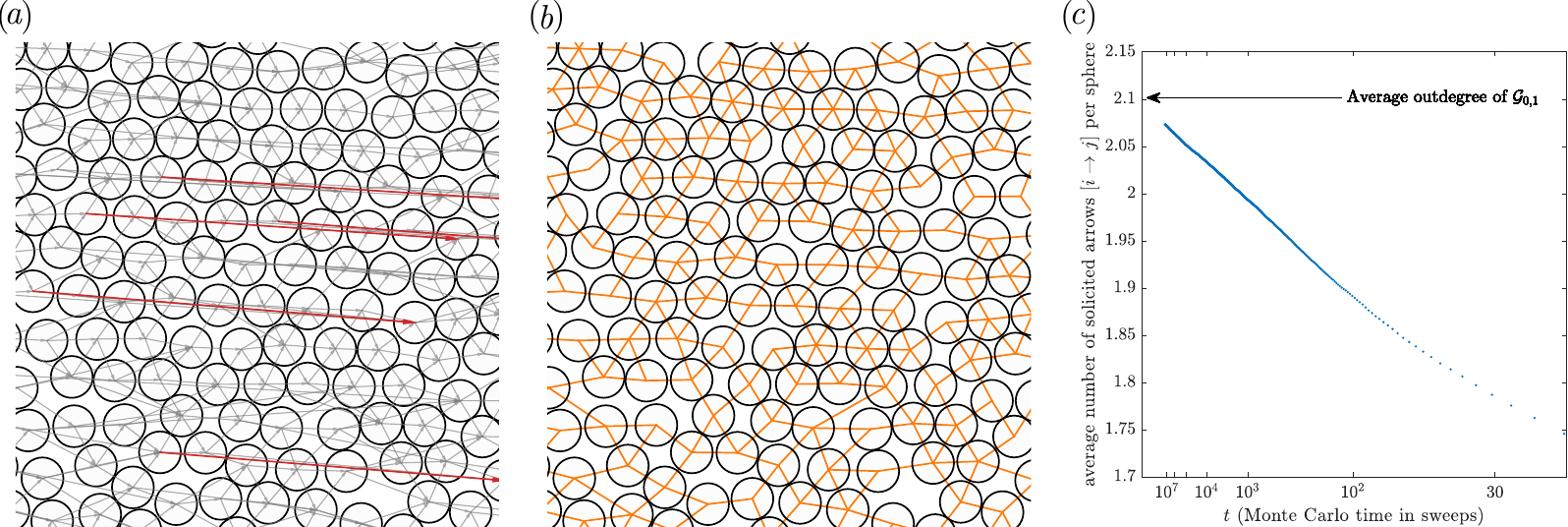}
\end{center}
\caption{ECMC constraint graphs for $\CCAL^{256}$ (see \sect{sec:Validation} for 
definition). \subcap{a} \Gthree for this configuration (detail), with 
highlighted non-local arrows. \subcap{b} $\GCAL^{256}$ (same detail),  obtained 
from \Gthree through fourth-order pruning followed by symmetrization. 
\subcap{c} Number of solicited arrows in  $\GCAL^{256}$ in 
a long cell-based ECMC run, compared to its average outdegree. }
\label{fig:ConstraintGraph}
\end{figure}

\subsection{Validation}
\label{sec:Validation}
Our programs apply to arbitrary density $\eta = N \pi \sigma^2/ L^2 $ and 
linear 
size $L$ of the periodic square box (with $N = M^2$). We provide sets of 
configurations and constraint graphs for validation and benchmarking. One such 
set consists in a 
configuration  $\CCAL^{256}$ at $M=256$, and for $\eta = 0.708$ and a
fourth-order symmetrized constraint graph $\GCAL^{256}$. Where 
applicable, the number of active spheres varies as $k = 1,2,4 \TO k^{\max}$ and 
the 
number of threads as $n_{\iota} =1,2,3 \TO n_{\iota}^{\max}$. For fixed $k$ and 
$n_{\iota}$, there are $n_{\text{run}}$ runs that vary $h$ and $h'$.

\paragraph{Constraint-graph validation}

Constraint graphs are generated in the \verb#Setup.sh# script. The 
\verb#GraphValidateCellECMC# test performs cell-based ECMC
derived from \verb#CellECMC.f90# ~\cite{Bernard2011,Engel2013}, where spheres 
are assigned to local cells and neighborhood-cell searches identify possible 
liftings. Cell-based ECMC must exclusively solicit liftings accounted for in 
$\GCAL$. The \verb#GraphValidateCellECMC# test also records the 
sweep (lifting per sphere) at which an arrow $\arrow{i}{j} \in \GCAL$ is first 
solicited in a lifting and 
compares the time evolution of the average number of solicited arrows with its 
average outdegree. The $\GCAL^{256}$ constraint graph passes the validation test 
with $t = \fpn{2}{7}$ sweeps. The outdegree of 
$\GCAL^{256}$ is $2.1$, and $98.7$ \%  of its arrows are solicited during the 
test. Logarithmic extrapolation 
(with $1/{(\ln{t})}^\alpha$, 
$\alpha=1.7$) suggests that $\GCAL^{256}$ essentially agrees with \Gmin (see 
\subfig{fig:ConstraintGraph}{c}). Use of $\GCAL^{256}$ rather than \Gthree 
speeds up ECMC, but further performance gains through additional pruning are 
certainly extremely limited.

\paragraph{Validation of \alggtwo{alg:LocalECMC}{alg:MultiCPP}}

Our implementations of \algtwo{alg:LocalECMC}{alg:MultiCPP} are modified as 
discussed in \rem{rem:PartialValidation}. Runs compute the set $\LCAL_{t^*}$ to 
the earliest horizon-violation time $t^*$ (with $t^* = h'$ if the run concludes 
successfully). The \verb#PValidateECMC.sh# test first advances 
$\CCAL^{256}= \CCAL_{t=0}$ to a random breakpoint $h$ (using \Lref).
Each test run is in the interval $[h,h']$, where $h'$ is randomly 
chosen. To pass the validation test, $\LCAL_{t^*}$ must for each run agree with 
\Lref (see 
\sect{sec:ComputerCode} for details of scripts used). \alg{alg:LocalECMC} passes 
the \verb#PValidateECMC.sh# test with  $n_{\text{run}} = 
\fpn{1}{3}$ for $k^{\max} = 8192$. 

Our x86 computer has two Xeon Gold 6230 CPUs with variable frequency from $2.1$ 
GHz to $3.9$ GHz, each with $20$ cores and $40$ hardware threads. We use OpenMP 
directives to restrict all threads to a single \xCPU. We consider again 
$\CCAL^{256}= \CCAL_{t=0}$ as the initial configuration, and then run the 
program from $h$ to $h'$. On our \xCPU, \alg{alg:MultiCPP} passes the 
\verb#CValidateECMC.sh# test with  $n_{\text{run}} = 
\fpn{1}{3}$, $k^{\max} = 8192$ and $n^{\max}_{\iota} = 40$.

On our \ACPU (Nvidia Jetson with Cortex A57 CPU (at $1.43$ GHz) with four cores 
and four hardware threads), we again consider $\CCAL^{256}= \CCAL_{t=0}$ 
as initial configuration. For the same system 
parameters as above, \alg{alg:MultiCPP} passes the \verb#CValidateECMC.sh# test 
with  $n_{\text{run}} = \fpn{1}{3}$, $k^{\max} = 8192$ and 
$n^{\max}_{\iota} = 4$. The ARM architecture allows dynamic 
re-ordering of operations, and the separate validation test
more severely scrutinizes thread interactions than for the \xCPU.

On both CPUs, \alg{alg:MultiCPP} passes the \verb#CValidateECMC.sh# 
test with the following choices of \verb#memory_order# directives: 
\begin{description} \item \verb#memory_order_relaxed#. This most permissive 
memory ordering of the  C++ memory model imposes no constraints on compiler 
optimization or dynamic re-ordering of operations by the processors, and only
guarantees the atomic nature of the CAS operation. Such 
re-orderings are more 
aggressive on \ACPUs than on \xCPUs. This memory ordering does not guarantee 
that the statements constituting the lock-less lock are executed as required 
(see \rem{rem:MemoryOrder}). 

\item \verb#memory_order_seq_cst# for all memory operations on the tag 
attribute. This directive imposes the sequential-consistency model (see 
\rem{rem:MemoryOrder}) for each access of the tag attribute. It
slows down the code by $40$\% compared to the \verb#memory_order_relaxed# 
directive. 

\item \verb#memory_order_acquire# on \verb#load#, \verb#memory_order_release# on 
\verb#store#. This directive implies \quot{acquire--release} semantics on the 
tag attribute. It
imposes a lock-free exchange at each operation on the tag attribute, 
so that all variables, including positions and local times, are synchronized 
between threads during tag access. CAS remains \verb#memory_order_seq_cst#. This 
directive maintains speed compared to
\verb#memory_order_relaxed#, 
yet provides better guarantees on the propagation of variable 
modification between threads. \verb#MultiThreadECMC.cc# compiles by default 
with this directive.

\end{description}

\subsection{Benchmarks for \alg{alg:MultiCPP} (x86 and ARM)}
\label{sec:Benchmarks}

\alg{alg:MultiCPP} is modified as discussed in \rem{rem:PartialValidation} 
(program execution continues in spite of horizon violations) and used for large 
values of $h'$. This measures the net cost of steady-state thread interaction, 
without taking into account thread-setup times. We report here on results of the 
\verb#BenchmarkECMC.sh# script for $\CCAL^{256}$ as an initial configuration 
and 
$\GCAL^{256}$ with  $k=80$ active spheres. The number of threads varies as 
$n_{\iota} =1,2,3 \TO n_{\iota}^{\max}$, with $n_{\text{run}}=20$. 

On our \xCPU (see \sect{sec:Validation}), the \verb#BenchmarkECMC.sh#
script is parame\-trized with $n_{\iota}^{\max} =40$. The benchmark speed 
increases roughly linearly up to $20$  threads (reaching a speed-up of $12$ for 
$20$ threads), and then keeps improving more slowly with a maximum for $40$ 
threads at a speed-up of $16$ and an absolute speed of $\sim \fpn{1.7}{12}$ 
events/hour (see \fig{fig:EPH_benchmark}). The variable frequency of Xeon 
processors under high load may contribute to this complex behavior. On a single 
thread, our program runs $2.4$ times slower than an unthreaded code, due to 
the 
eliminated overhead from threading constructs. The original 
\verb#CellECMC.f90# cell-based production code generates \fpn{3}{10} 
events/hour. The use of a  constraint graph, rather than a cell-based search, 
thus improves performance by almost an order of magnitude, if the set-up of 
$\GCAL$ is not accounted for.

On our \ACPU, the \verb#BenchmarkECMC.sh# script is parametrized with 
$n_{\iota}^{\max} =4$. The benchmark speed increases as the number of threads, 
reaching a speed-up of $4.0$ for $n_{\iota}=4$. The absolute speed is about 
seven times smaller than for our \xCPU for a comparable number of threads, as 
may be expected for a low-power processor designed for use in mobile phones.

\draftfigure[10cm]{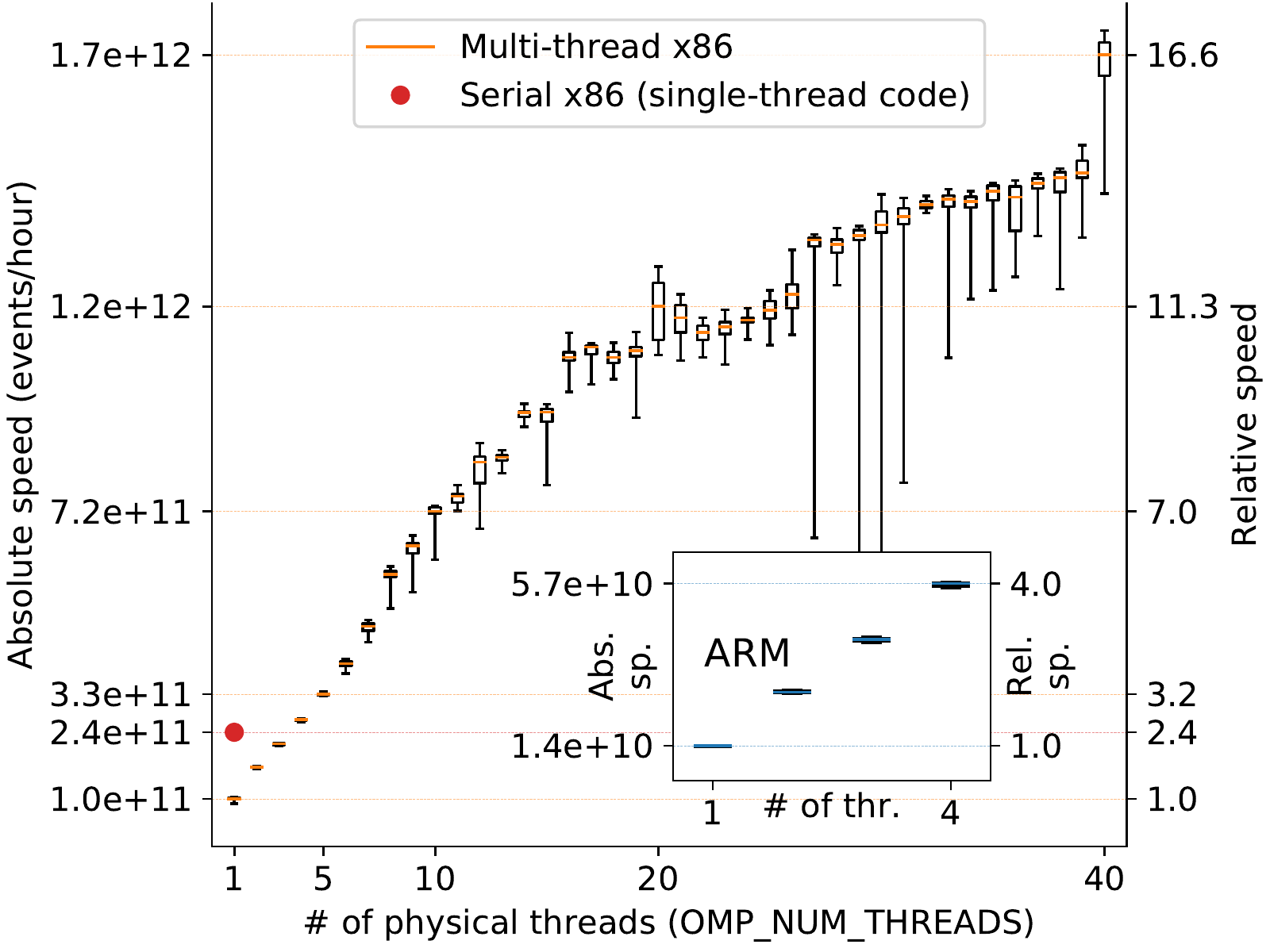}
{Output of the \texttt{BenchmarkECMC.sh} script  for \alg{alg:MultiCPP}
(five-number summary of $20$ runs) for $k=80$
on an \xCPU with $20$ cores and $1 \TO 40$ threads, and for 
serial code that processes active chains 
sequentially. Inset: Output of the script for 
a four-core 
ARM CPU.}
{fig:EPH_benchmark}

\subsection{Birthday problem, full ECMC, multi-CPU extensions}
\label{sec:Extensions}

In this section, we treat some practical aspects for the use of
\alg{alg:MultiCPP}. 

\paragraph{Birthday problem}
We analyze multithreaded ECMC in terms of the 
(generalized) birthday problem, which considers the probability $p$ that two 
among $\kbirth$ integers (modeling individuals) sampled from a discrete uniform 
distribution in the set $\SET{1,2 \TO \Nbirth}$ (modeling birthdays) 
are the same. For large $\Nbirth$, $p \sim [1-\exp{   (-\kbirth^2/ (2 \Nbirth)  
) } ]$~\cite{Mathis1991}, which is small if $\kbirth \lesssim \sqrt{\Nbirth}$. 
At 
constant density $\eta$, sphere radius $\sigma$, velocity $\vvactive$, and time 
interval $h' - h$, each active chain $\iota$ is restricted to a region of 
constant area, whereas the total area of the simulation box is $\eta N$. We may 
suppose that the $k$ active spheres are randomly positioned in the simulation 
box broken up into a grid of $\propto N$ constant-area cells. For $k \lesssim 
\sqrt{N}$, we expect the probability that one of these cells contains two active 
spheres to remain constant for $N \to \infty$, and therefore also the 
probability of an update-order violation for constant $h'-h$.

\paragraph{Restarts}
Our algorithms reproduce output of \alg{alg:GlobalECMC} only if they do not 
abort. In production code, the effects of horizon violations will have to be 
repaired. Two strategies appear feasible. First, the algorithm may restart the 
run from a copy of \lifted{\CCAL_h}{\ACAL_h} at the initial breakpoint $h$, and 
choose a smaller breakpoint $h''$, for example the time of 
abort. The successful termination of this restart is not guaranteed, as the 
individual threads may organize differently. Second, the time evolution may be 
reconstructed from $\LCAL_{t^*}$ to the earliest horizon-violation time $t^*$ 
(see \rem{rem:PartialValidation}), and $t^*$ may then be used as the subsequent 
initial breakpoint. Besides an efficient restart strategy, a multithreaded 
production code will also need an efficient parallel algorithm for computing 
$\GCAL$ after a change of $\vvactive$.

\paragraph{Multi-CPU implementations}
\alg{alg:MultiCPP} is spelled out for a single shared-memory CPU
and for threads that may access attributes of all spheres (see statement 
$10\thread$ in \alg{alg:MultiSeqConsist} and \rem{rem:Necklace}). However, 
thread interactions are local and immutable in between breakpoints (as evidenced 
by the constraint graphs). This invites generalizations of the algorithm to 
multiple CPUs (each of them with many threads). Most simply, two CPUs could 
administer disjoint parts of the constraint graph, for example with interface 
vertices doubled up on both of them (see \fig{fig:ConstraintCutup}). In this 
way, an active sphere arriving at an interface would simply be copied out to 
the neighboring CPU. The generalization to multiple CPUs appears 
straightforward.  \draftfigure[6cm]{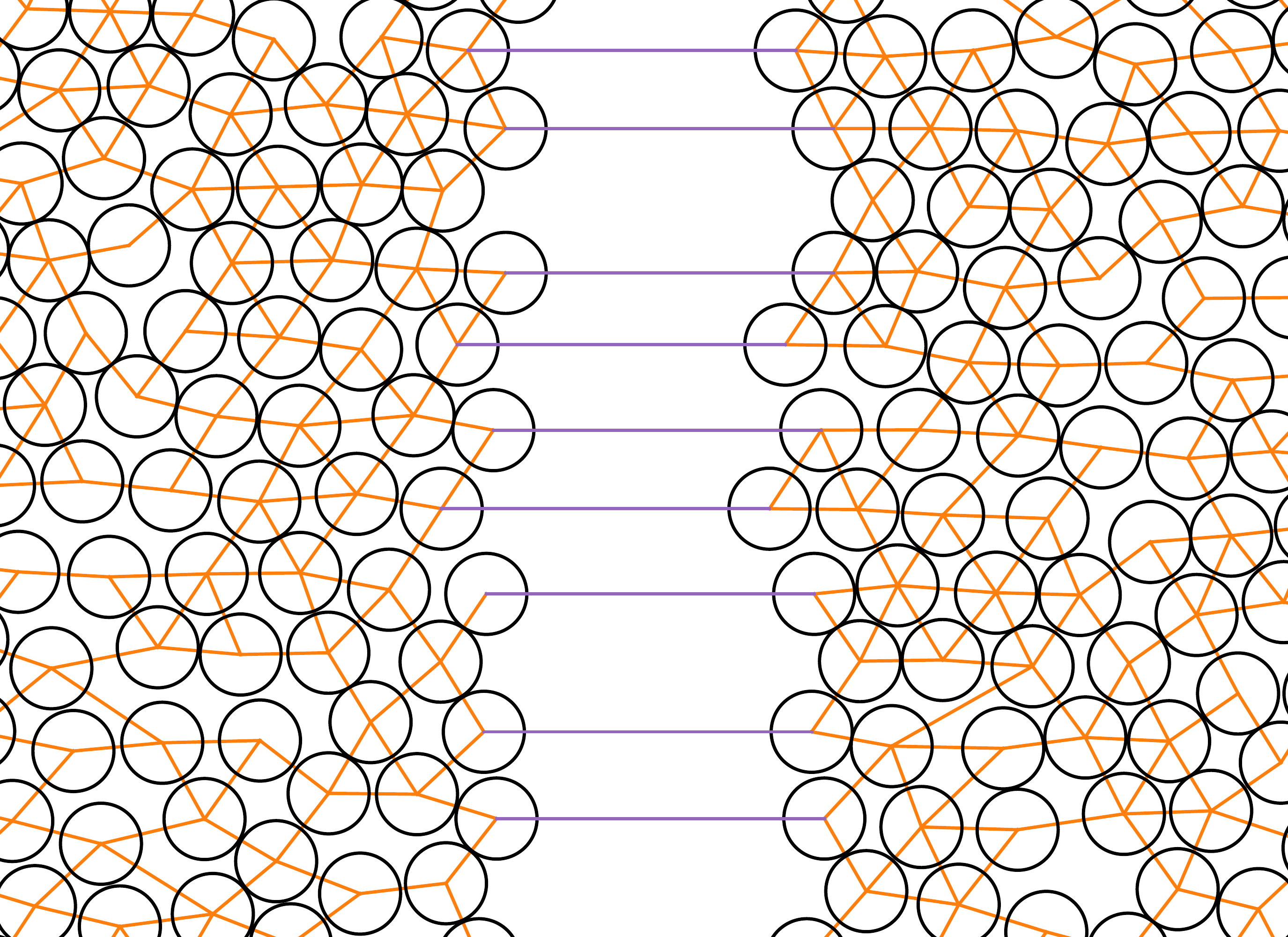}{A constraint graph 
doubled up for a multi-CPU implementation of \algg{alg:MultiCPP} (detail of 
$\CCAL^{256}$ configuration shown). Interface vertices appear on both 
sides.}{fig:ConstraintCutup}

\section{Available computer code}
\label{sec:ComputerCode}

All implemented algorithms and used scripts that are made available on GitHub in 
\verb#ParaSpheres#, the public repository, which is part of a public GitHub 
organization.\footnote{The organization's url is  
\url{https://github.com/jellyfysh}.} Code is made 
available under the GNU GPLv3 license (for details see the \verb#LICENSE# file).

The repository can be forked (that is, copied to an outside user's own public 
repository) and from there studied, modified and run in the user's local 
environment. Users may contribute to the \verb#ParaSpheres# project  
\emph{via} pull requests (see the \verb#README.md# and \verb#CONTRIBUTING.md# 
files for instructions and guidelines). All communication (bug reports, 
suggestions) take place through GitHub \quot{Issues}, that can be opened in the 
repository by any user or contributor, and that are classified in GitHub 
projects.

\paragraph{Implemented algorithms}

The following programs are located
in the directory tree under their  language (\verb#F90#, \verb#Python# or 
\verb#CPP#) and in similarly named subdirectories, that
all contain \verb#README# files for further details.
Some of the longer programs are split into modules.\\
\\
\begin{tabular}{ll}
Code/Directory   & Algorithm / Usage \\ \hline
\verb#CellECMC.f90# & Cell-based production ECMC~\cite{Bernard2011}\\
\verb#GenerateG3.py# &       Generate \Gthree 
(\sect{sec:ConstraintGraph})\\ 
\verb#PruneG.py# & Prune $\GCAL$ (\sect{sec:ConstraintGraph})\\ 
\verb#GraphValidateCellECMC.cc# & Validate $\GCAL$ against cell-based ECMC \\ 
\verb#GlobalTimeECMC.py# & \algg{alg:GlobalECMC} 
(\sect{sec:GlobalContinuousECMC}) \\ 
\verb#SingleThreadLocalTimeECMC.py# & \algg{alg:LocalECMC} 
(\sect{sec:LocalContinuousECMC}) \\ 
\verb#SequentialMultiThreadECMC.py# & \algg{alg:MultiSeqConsist} 
(\sect{sec:MultiSeqConsist}) \\
\verb#MultiThreadECMC.cc# & \algg{alg:MultiCPP} (\sect{sec:MultiCPP})
\end{tabular}

\paragraph{Scripts and validation suites}

The \verb#Scripts# directory provides the following
bash scripts to compile and run groups of 
programs and 
to reproduce all our results:\\
\\
\begin{tabular}{ll}
Script        & Summary of usage \\ \hline
\verb#Setup.sh# & Prepare $\CCAL_{t=0}$, $\GCAL$, \Lref \\
\verb#SequentialC4.sh# & Test suite for \algg{alg:MultiSeqConsist} with 
$N=4$\\
\verb#SequentialC5.sh# & Test suite for \algg{alg:MultiSeqConsist} with $N=5$
(see \rem{rem:MemoryOrder})\\ 
\verb#Ordering.sh# & Test suite for \algg{alg:MultiCPP} with $N=5$
(see \rem{rem:MemoryOrder})\\ 
\verb#ValidateG.sh# & Validate constraint graph \\
\verb#PValidateECMC.sh# & Validate \algg{alg:LocalECMC} against \Lref  \\
\verb#CValidateECMC.sh# & Validate \algg{alg:MultiCPP} against \Lref \\
\verb#BenchmarkECMC.sh# & Benchmark \verb#MultiThreadECMC.cc#, generate
\fig{fig:EPH_benchmark} \\
\end{tabular} \\

In the \verb#Setup.sh# script, \verb#CellECMC.f90# first produces an 
equilibrated sample $\CCAL_{0}$. It then generates \Gthree with  
\verb#GenerateG3.py#, and runs \verb#PruneG.py# to output $\GCAL$. Finally, it 
runs \verb#GlobalTimeECMC.py# for each set $\ACAL_{0}$, in order to generate 
generate several \Lref. The \verb#ValidateG.sh# script runs cell-based ECMC, 
using \verb#GraphValidateCellECMC.cc#, and verifies that all liftings are 
accounted for in $\GCAL$. It also tracks the solicitation of arrows as a 
function of time. The \verb#PValidateECMC.sh# script validates 
\verb#SingleThreadLocalTimeECMC.py#  by comparing the sets of liftings with
\Lref from \verb#Setup.sh#. The \verb#CValidateECMC.sh# script does the 
same for \verb#MultiThreadECMC.cc# The \verb#BenchmarkECMC.sh# script benchmarks 
\verb#MultiThreadECMC.cc# for different numbers of threads. The test suites are 
concerned with small-$N$ configurations.  

\section{Conclusions and outlook}
\label{sec:ConclusionsOutlook}
In this paper, we presented an event-driven multithreaded ECMC algorithm for 
hard spheres which enforces thread synchronization at infrequent breakpoints 
only. Between breakpoints, spheres carry and update local times. Possible 
inconsistencies are locally detected through a horizon condition.  Within ECMC, 
our method avoids the scheduling problem that has historically plagued 
event-driven molecular dynamics. This is possible because in ECMC only few
spheres  move  at any moment, and all have the same velocity. 
Conflicts are thus exceptional, and little information 
is exchanged between threads. We relied the generalized birthday problem 
to 
show that our algorithm remains viable up to a number of threads that grows as 
the square root of the number of spheres, a setting relevant for the 
simulation of millions of spheres for modern commodity servers with $ \sim 100$ 
threads.  The mapping of \alg{alg:MultiSeqConsist} onto an absorbing Markov 
chain allowed us to prove its correctness (for a given lifted initial 
configuration), and to rigorously analyze side effects of code re-orderings in 
the multithreaded  C++ code.

Our algorithm is presently implemented between two global breakpoint times, 
where it achieves considerable speed-up with respect to sequential ECMC.  
Generalization for a full practical multithreaded ECMC code that greatly 
outperforms cell-based algorithms appear within reach. It is still a challenge 
to understand whether multithreaded ECMC applies to general interacting-particle 
systems.

\section*{Acknowledgements}

W.K. acknowledges support from the Alexander von Humboldt Foundation. We 
thank E.~P.~Bernard for allowing his original hard-sphere ECMC production code 
to be made available. 

\bibliographystyle{elsarticle-num}
\bibliography{General,Para_spheres}
     
\end{document}